\documentclass{ifacconf}
\usepackage{natbib}   
\usepackage{graphicx}
\usepackage{textcomp}
\usepackage{amsmath,amssymb,bm,bbm,mathrsfs,amscd}
\usepackage{calc}
\usepackage{xcolor}
\usepackage{tikz}
\usetikzlibrary{arrows.meta, positioning}
\usepackage{dsfont}
\usepackage{epstopdf}
\usepackage{amsfonts}
\usepackage{rotating}
\usepackage{mathtools}
\usepackage{subfig}
\usepackage{enumerate,pgfplots}
\pgfplotsset{compat=1.18}

\newcommand{\ba}{\begin{array}}
\newcommand{\ea}{\end{array}}

\newcommand{\be}{\begin{equation}}
\newcommand{\ee}{\end{equation}}

\newcommand{\ds}{\displaystyle}

\newcommand{\mc}{\mathcal}

\newcommand{\R}{\mathbb{R}}

\def\R{\mathbb{R}}

\def\diag{{\rm diag}\,}
\begin{document}
\begin{frontmatter}
\title{On a  Co-evolving Opinion-Leadership Model in Social Networks} 

\thanks[footnoteinfo]{This work was partially supported by the Wallenberg AI, Autonomous Systems and Software Program (WASP) funded by the Knut and Alice Wallenberg Foundation.}

\author[KTH]{Martina Alutto} 
\author[Polito]{Lorenzo Zino} 
\author[KTH]{Karl H. Johansson}
\author[KTH]{Angela Fontan}

\address[KTH]{Department of Decision and Control Systems, School of Electrical Engineering and Computer Science, \\KTH Royal Institute of Technology, Stockholm, Sweden \\(e-mails: \{alutto; angfon; kallej\}@kth.se)}
\address[Polito]{Department of Electronics and Telecommunications, \\Politecnico di Torino, Turin,
Italy (email: lorenzo.zino@polito.it)}

\begin{abstract}
    Leadership in social groups is often a dynamic characteristic that emerges from interactions and opinion exchange. Empirical evidence suggests that individuals with strong opinions tend to gain influence, at the same time maintaining alignment with the social context is crucial for sustained leadership. Motivated by the social psychology literature that supports these empirical observations, we propose a novel dynamical system in which opinions and leadership co-evolve within a social network. Our model extends the Friedkin--Johnsen framework by making susceptibility to peer influence time-dependent, turning it into the leadership variable. Leadership strengthens when an agent holds strong yet socially aligned opinions, and declines when such alignment is lost, capturing the trade-off between conviction and social acceptance. After illustrating the emergent behavior of this complex system, we formally analyze the coupled dynamics, establishing sufficient conditions for convergence to a non-trivial equilibrium, and examining two time-scale separation regimes reflecting scenarios where opinion and leadership evolve at different speeds.
\end{abstract}


\end{frontmatter}

\section{Introduction}
Leadership in social groups frequently emerges endogenously from interpersonal interactions rather than being externally imposed~\citep{Nakayama2019}. Across social psychology and cognitive science, leadership is a social role that emerges when people's views strongly align with those of their peers and when they are seen as important and influential in their relational environment \citep{Surowiecki2004,isenberg1986group}. People with strong or extreme opinions frequently become more prominent in group discussions because their arguments are more confidently and cogently presented, which helps them influence group decisions and judgments \citep{VanSwol01072009}. Experimental evidence supports this phenomenon, showing that individuals with strong opinions are typically less susceptible to social influence, which manifests as persistent commitment to their views and increased perceived leadership \citep{brandt2015unthinking}. However, extremity alone does not guarantee long-term influence. Recent research suggests that leadership is a dynamic attribute shaped by ongoing relational adaptation, as even strongly polarized individuals may gradually converge towards moderate positions when exposed to persistent social pressure \citep{klein2023respondents}. More generally, people tend to reduce the dissonance between their own beliefs and those of their peers due to the enduring human need for social validation, rooted in relational consistency and epistemic certainty \citep{hillman2023social}. In summary, a fundamental mechanism for comprehending the emergence and evolution of leadership in complex social systems is the conflict between convictions and conformity in social communities, within an opinion formation and sharing process.

Opinion dynamics models provide a mathematical framework for studying belief updating driven by social interactions. Foundational contributions, such as the French--DeGroot model \citep{DeGroot1974} and the Friedkin--Johnsen model \citep{Friedkin1990}, describe belief evolution through weighted averaging processes capturing consensus formation and stable disagreement. More recently, the incorporation of reflected appraisal mechanisms has enabled the endogenous evolution of social power. According to this theory, individuals adjust their influence based on how much their expressed opinions affect those of their neighbors \citep{friedkin2011formal}. This has fueled a growing research stream on the co-evolution of opinions and influence in networked societies \citep{jia2015opinion,mirtabatabaei2014reflected,ohlin2022achieving,tian2021social,wang2021achieving,wang2022social,wang2024social,Kang2022,Liu2024}. These contributions show that influence is not static, but shaped dynamically through feedback loops linking susceptibility, opinion centrality, and network structure.
While these works deepen the understanding of influence formation, the emergence of leadership has often been modeled by assuming exogenous (and often constant) differences in stubbornness or authority. In contrast, empirical literature suggests that leadership is intrinsically tied to and dynamically shaped by the interplay between opinion strength and social validation \citep{VanSwol01072009, berglund2006noise}. Mathematical models capable of endogenously capturing such complex dynamics in an analytically-tractable formalism remain limited, calling for the development of new modeling approaches to account for this important mechanism.

To address this gap, we propose a novel dynamical framework consisting in a coupled system of nonlinear ordinary differential equations (ODEs), in which opinions and leadership co-evolve within a social network. In our formulation, an individual's leadership is shaped by two contrasting mechanisms, constantan with the social psychology literature discussed in the above. In particular, leadership rises when the individual holds strong opinions, and it declines when attitudes become excessively misaligned with those of their reference group (e.g., friends, colleagues, or online communities). 
Our approach can be interpreted as a modification of the Taylor's model \citep{taylor1968towards}, the continuous time counterpart of the Friedkin–Johnsen model ~\citep{Friedkin1990}, in which the susceptibility to peer influence is no longer static but evolves over time and becomes the variable representing leadership. Indeed, unlike existing approaches that prescribe influential individuals a priori, we represent leadership as a state variable that strengthens when an agent holds a strong opinion that remains aligned with the prevalent social context, and weakens when such alignment is lost. This mechanism captures the fundamental trade-off between conviction and social acceptance that governs informal leadership formation. 
 
The main contribution of this paper, besides the formalization of such a modeling framework, is threefold. First, we present some case studies on small networks to illustrate the model and discuss its main properties. Second, we provide a formal characterization of the coupled dynamics and show sufficient conditions ensuring convergence to a non-trivial equilibrium, using the center manifold method and the Banach fixed point theorem. Third, we examine two time-scale separation regimes, where opinion updating and leadership evolution occur at different rates, capturing scenarios commonly observed in social systems. For these scenarios, we first characterize the dynamics of the faster variable, and then we establish sufficient conditions guaranteeing convergence of the slower variable to a unique equilibrium. Numerical simulations illustrate the results and system behavior.

The rest of the paper is organized as follows. Section~\ref{sec:model} introduces the opinion-leadership model, while Section~\ref{sec:examples} provides intuition on its dynamics through illustrative examples. Section~\ref{sec:stability} presents the analysis of equilibrium points and their stability, while Section~\ref{sec:timescale} discusses the time-scale separation analysis. Finally, Section~\ref{sec:conclusion} concludes the paper and outlines directions for future research.

\section{Model Description}\label{sec:model}
\subsection{Notation}
We denote by $\R$ and $\R_{+}$ the sets of real and nonnegative real numbers, respectively, while $\R_{+}^{n \times n}$ indicates the set of real matrices with dimension $n \times n$ and nonnegative entries. 
The all-1 vector and the all-0 vector are denoted by $\boldsymbol{1}$ and $\boldsymbol{0}$ respectively. The identity matrix and the all-0 matrix are denoted by $I$ and $\mathbb{O}$, respectively.
The transpose of a matrix $A$ is denoted by $A^T$. 
For $x$ in $\R^n$, let $||x||_1=\sum_i|x_i|$ and $||x||_{\infty}=\max_i|x_i|$ be its $l_1$- and $l_{\infty}$- norms, while $\mathrm{diag}(x)$ denotes the diagonal matrix whose diagonal coincides with $x$. 
For an irreducible matrix $A$ in $\R_+^{n\times n}$, we let $\rho(A)$ denote the spectral radius of $A$. 
Inequalities between two vectors $x$ and $y$ in $\R^n$ are meant to hold true entry-wise, i.e., $x \le y$ means that $x_i\le y_i$ for every $i$, whereas $x< y$ means that $x_i< y_i$ for every $i$, and $x\lneq y$ means that $x_i\le  y_i$ for every $i$ and $x_j<y_j$ for some $j$. Analogous holds for the matrix. Unless differently specified, we use the notation $\sum_j$ to denote $\sum_{j=1}^n$.

We consider a social network composed of $n$ agents. The interactions among agents are represented by a matrix $W \in {\mathbb{R_+}}^{n \times n}$, which is assumed to be row-stochastic (i.e., $\sum_j W_{ij} = 1$, for all $i$), so that $W_{ij}$ quantifies the weight that agent $i$ assigns to the opinion of agent $j$.  
Each agent $i \in \{1, \dots, n\}$ is characterized by two dynamical variables:  
\begin{itemize}
	\item An opinion variable $x_i(t) \in \mathbb{R}$, representing the belief of agent $i$ on a certain topic;  
	\item A leadership variable $y_i(t) \in [0,1]$, representing the degree of leadership of agent $i$, 
    with larger values of $y_i(t)$ denoting that agent $i$ has a larger tendency to lead the group at time $t$. 
\end{itemize}

The joint evolution of opinions and leadership is governed by the following system of ODEs:
\be \label{eq:leader-model}
\begin{cases}
	\dot{x}_i = \ds - \Big( x_i - \sum\nolimits_{j} W_{ij} x_j \Big)(1 - y_i) - y_i\big(x_i - x_i(0)\big) \\
	\dot{y}_i = \ds \alpha (1 - y_i) x_i^2 - \beta y_i \Big( x_i - \sum\nolimits_{j} W_{ij} x_j \Big)^2,
\end{cases}
\ee
for all $i\in\{ 1,\dots,n\}$, where $\alpha, \beta > 0$ are scalar parameters.

Before analyzing the dynamics, we present a brief discussion to elucidate its mechanisms. From \eqref{eq:leader-model}, we observe that the opinion dynamics is driven by two main mechanisms. The first term, weighted by $(1 - y_i)$, captures the tendency of followers to conform to the average opinion of their neighbors, which is a well-known phenomenon in social psychology~\citep{Asch1952}, and it is the key mechanism in the classical French--DeGroot opinion dynamic model~\citep{DeGroot1974}. In contrast, the second term reflects the stubbornness of leaders, who resist changes and tend to remain consistent with their initial opinion $x_i(0)$. In other words, for a fixed leadership value $y_i$, the first equation in \eqref{eq:leader-model} can be interpreted as a continuous-time version of the classical Friedkin--Johnsen model ~\citep{Friedkin1990, taylor1968towards}.
Instead, in this model, the agents' leadership co-evolve with their opinion, shaped by two opposing mechanisms: the first term accounts for the tendency of agents with strong opinions to emerge as leaders, as supported by the social psychology literature~\citep{VanSwol01072009,brandt2015unthinking}. Here, the parameter $\alpha$ can be interpreted as a \emph{leadership growth rate}, controlling how quickly strong opinions translate into increased leadership. Conversely, the second term represents the loss of leadership due to social misalignment: agents whose opinions deviate significantly from those of their neighbors lose influence, as suggested by the social verification theory~\citep{hillman2023social}. The parameter $\beta$ therefore measures the \emph{sensitivity of leadership to disagreement}, indicating how strongly misalignment penalizes an agent’s leadership. We refer to system~\eqref{eq:leader-model} as the \emph{opinion-leadership model}.

If we conveniently gather the opinion and leadership of all agents in two vectors $x=[x_1 \cdots x_n]^\top$ and  $y=[y_1 \cdots y_n]^\top$, respectively, the system in \eqref{eq:leader-model} can be rewritten compactly as:
\be \label{eq:leader-model-compact}
\begin{cases}
	\dot{x} = - (I - Y) (I - W) x - Y (x - x(0)), \\
	\dot{y} = \alpha\, (I - Y) \mathrm{diag}(x)\, x - \beta\, Y\mathrm{diag}(( I - W)x)\, ( I - W)x,
\end{cases}
\ee
where $Y=\mathrm{diag}(y)$ is the diagonal matrix whose diagonal entries are the leadership variables of the agents.

The next result establishes well-posedness of the model by showing that the set
$$ \Delta := [-1,1]^n \times [0,1]^n$$
is positively invariant for the opinion-leadership coevolution model, i.e., if $(x(0),y(0))$ is in $\Delta$, then $(x(t), y(t))$ will be in $\Delta$ for all $t \geq 0$.
\begin{lem}
	For any initial condition $(x(0),y(0)) \in \Delta$ there exists a unique solution $(x(t),y(t))$ of \eqref{eq:leader-model} defined for all $t \ge 0$. Moreover, $\Delta$ is positively invariant for \eqref{eq:leader-model}. 
\end{lem} 
\begin{pf}
	Since the vector field of \eqref{eq:leader-model} is Lipschitz with respect to $(x,y)$, local existence and uniqueness for the Cauchy problem are guaranteed by the Picard–Lindelöf theorem \citep{hale}.  
	The domain $\Delta$ is compact and convex, hence Nagumo’s theorem can be applied \citep[Theorem 3.1]{blanchini1999set}. To verify invariance, we check the direction of the vector field at the boundaries of $\Delta$. We observe that, for all $i=1,\dots,n$, $\dot x_i \leq 0$ if $x_i = 1$ and $\dot x_i \geq 0$ if $x_i = -1$. Similarly, $\dot y_i \leq 0$ if $y_i = 1$ and $\dot y_i \geq 0$ if $y_i = 0$. 
	It follows that any solution $(x(t),y(t))$ with $(x(0),y(0))\in \Delta$ remains in $\Delta$ for all $t \geq 0$, which implies that $\Delta$ is a positively invariant set  for \eqref{eq:leader-model}. \qed
\end{pf}


\section{Illustrative Examples}\label{sec:examples}
In this section, we analyze different scenarios to highlight the emergent phenomena that the model can capture and reproduce.

\subsection{Three-node network}
We consider the opinion-leadership model~\eqref{eq:leader-model} in the case of a network formed by $3$ interacting agents, starting with polarized opposite opinions for agents $1$ and $3$ ($x_1(0) = -1$ and $x_3(0)=1$), and neutral opinion for agent $2$ ($x_2(0)=0$), which we also denote as the central node in what follows. We assume a symmetric interaction matrix of the form
$$W_{ij}=\left\{\begin{array}{cl}a&\text{if }i=j,\\
\frac{1-a}{2}&\text{if }i\neq j,
\end{array}\right.$$
where the constant $a\in(0,1)$ captures the relative weight of the self-loop with respect to the other entries.

Due to the symmetry of $W$, it can be shown (proof and numerical simulations are here omitted for lack of space) that the influences of the two initially polarized agents acting on node $2$ are perfectly balanced.  
As a result, the central agent remains fixed at its initial neutral opinion, i.e., $x_2^* = 0$. 
Interestingly, we further observe that the symmetry implies that the leadership of the two polarized agents, even if initially different, eventually converges to the same value (see Fig.~\ref{fig:network3-sim_a} for a numerical simulation). This value is larger than the one of the neutral agent, due to the first term in the leadership dynamics that captures the real-world emergence of leaders among those with stronger opinions. 

\begin{figure}
    \subfloat[]{\hspace{-3pt}\includegraphics[width=0.52\linewidth]{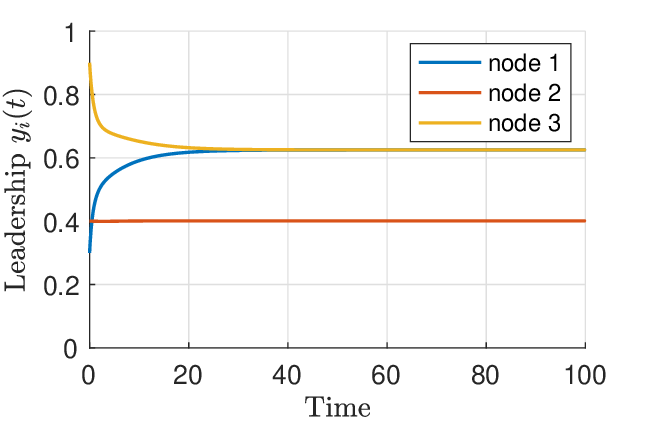}\label{fig:network3-sim_a}}
    \subfloat[]{\hspace{-6pt}\includegraphics[width=0.52\linewidth]{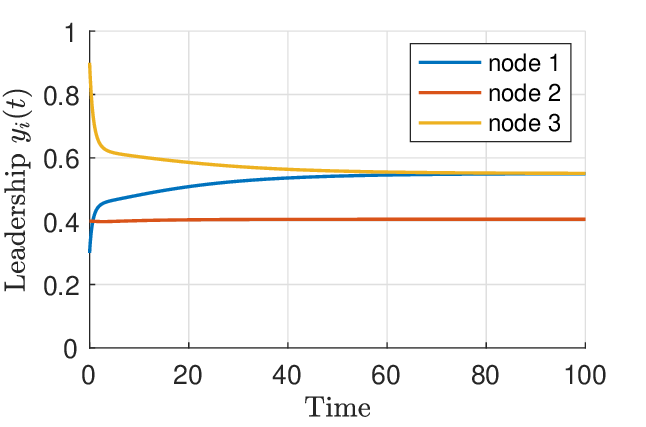}\label{fig:network3-sim_b}}\\    
    \subfloat[]{\hspace{-3pt}\includegraphics[width=0.52\linewidth]{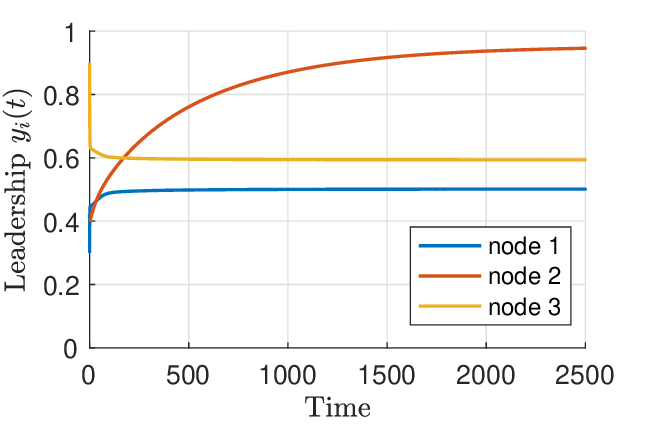}\label{fig:network3-sim_c}}
    \subfloat[]{\hspace{-6pt}\includegraphics[width=0.52\linewidth]{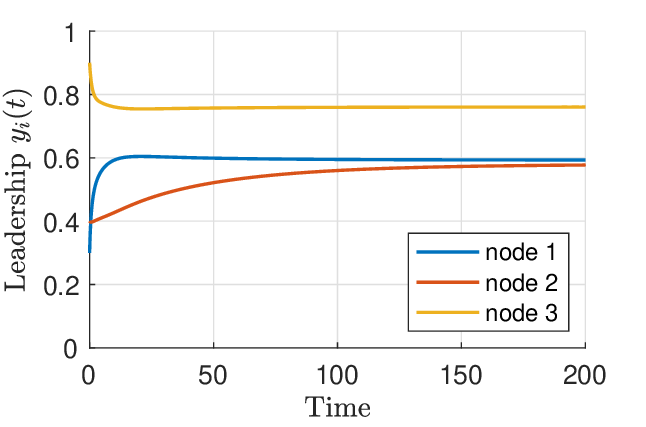}\label{fig:network3-sim_d}}
    \caption{Numerical simulation of the three-agent opinion-leadership model with $\alpha = 0.6$ and $\beta =1$. (a) Symmetric case: agent $2$ remains neutral while agents $1$ and $3$ converge to moderated polarized opinions. (b)–(d) Perturbed scenarios showing how asymmetry in the opinion initialization or interaction strengths breaks polarization symmetry and modifies leadership emergence. }
    \label{fig:network3-sim}
\end{figure}

We now investigate perturbations of this symmetric setting.
First, we consider a weaker mutual influence between nodes $1$ and $3$, i.e., we reduce the entries $W_{13} = W_{31} = c < \frac{1-a}{2}$.
In this case, illustrated in Fig.~\ref{fig:network3-sim_b}, the leadership levels of agents $1$ and $3$ become slightly smaller. With reduced mutual reinforcement, agents polarized agents experience less opinion confirmation and thus exhibit reduced leadership.

Second, we assume that the central node has a small, but non-zero, initial condition, i.e., $x_2(0) = 0.05$. In this scenario, depicted in Fig.~\ref{fig:network3-sim_c}, the leadership of nodes $1$ and $3$ still converge at a similar rate toward equilibrium but no longer symmetrically.
Hence, the agent with positive initial opinion (node $3$) eventually emerges as a stronger leader than agent $1$. Interestingly, we observe that agent $2$ exhibits a much slower convergence because their initial state is close to $0$ and due to the symmetric of matrix $W$, the disagreement term driving the leadership dynamics is small in magnitude. This leads to a gradual but steady leadership update, eventually becoming dominant in the social network. 

Third, we consider a case where node $2$ has a negative initial opinion as $x_2(0)=-0.05$ and a stronger influence from node $1$, i.e., let $W_{12} > W_{32}$.
In this case, as shown in Fig.~\ref{fig:network3-sim_d}, the absence of symmetries leads to a more complex emergent behavior, whereby agent $1$ receives a stronger reinforcing influence from the central node, resulting in a larger leadership compared to agent $3$.

\subsection{Impact of model parameters}

\begin{figure*}
    \centering
    \subfloat[Opinions, $\beta=0.1$]{\includegraphics[width=0.25\linewidth]{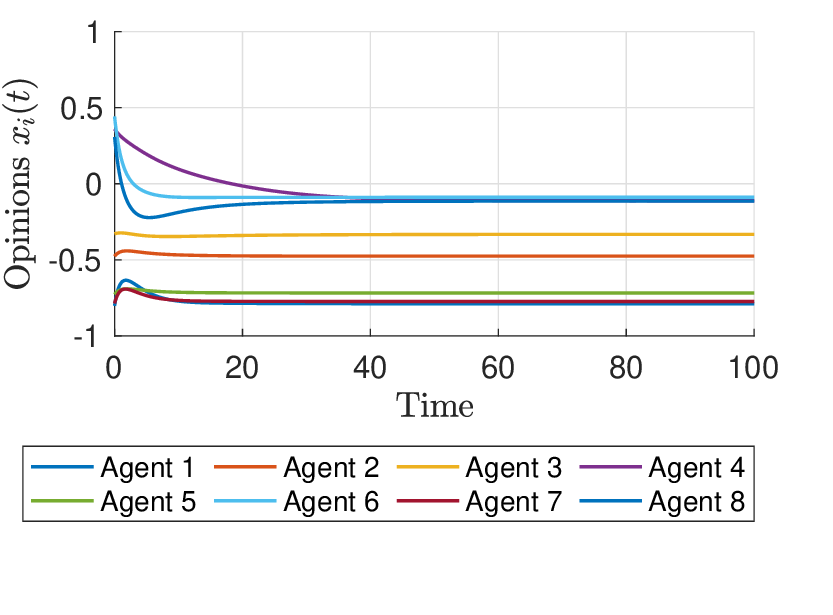}\label{fig:examplea}}
    \subfloat[Leadership, $\beta=0.1$]{\hspace{-6pt}\includegraphics[width=0.25\linewidth]{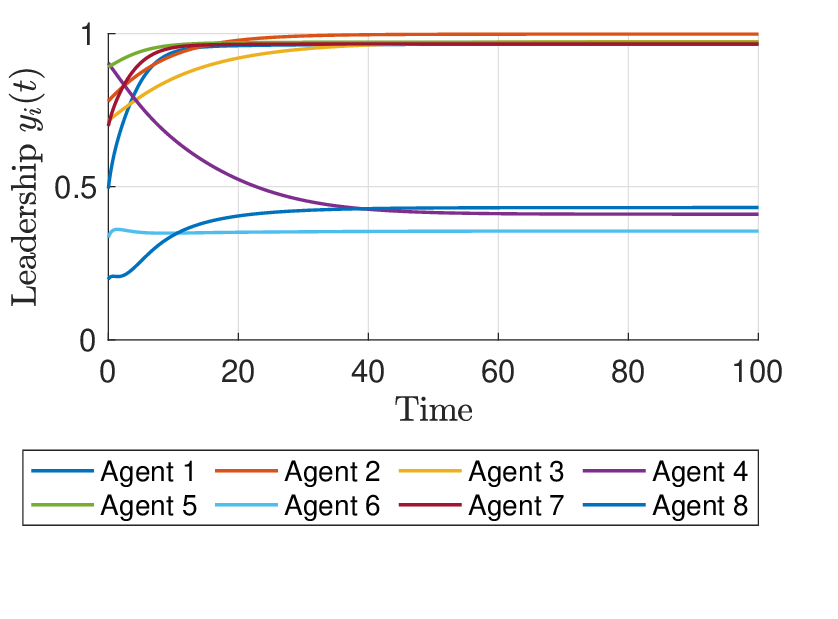}\label{fig:exampleb}}
      \subfloat[Opinions, $\beta=1$]{\includegraphics[width=0.25\linewidth]{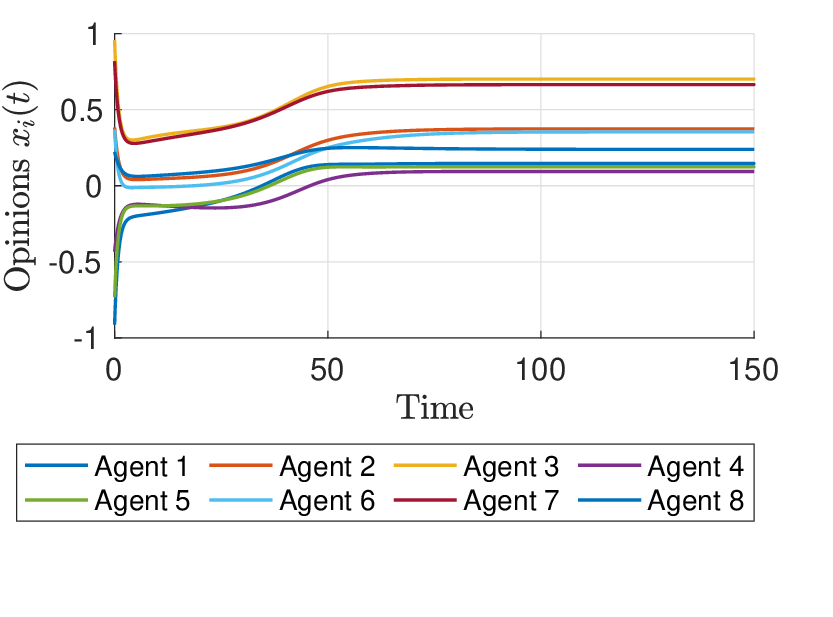}\label{fig:examplec}}
    \subfloat[Leadership, $\beta=1$]{\hspace{-6pt}\includegraphics[width=0.25\linewidth]{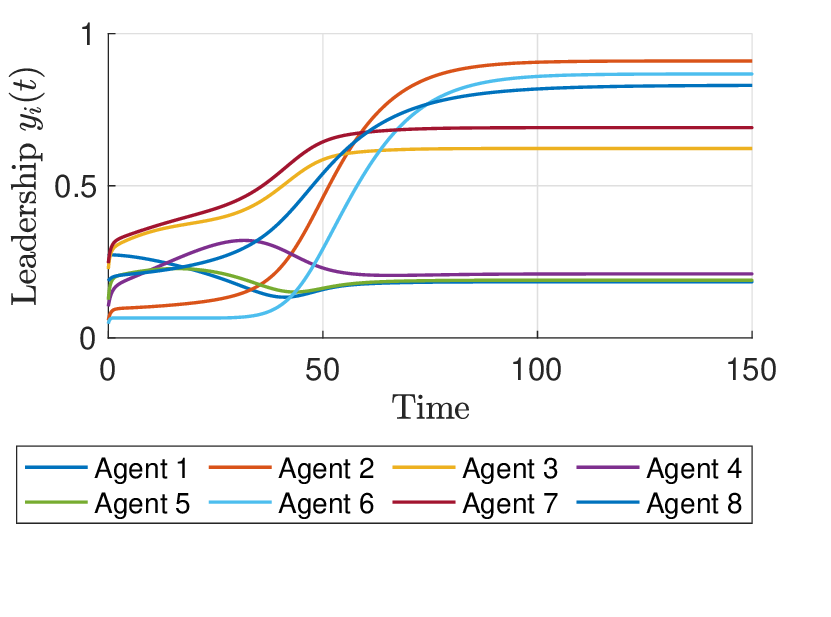}\label{fig:exampled}}
    \caption{Simulation of the opinion-leadership model~\eqref{eq:leader-model} with $8$ agents with $\alpha = 0.6$ and (a,b) $\beta = 0.1$ or (c,d) $\beta = 1$. }
    \label{fig:example1}
\end{figure*}

We consider now a population of $8$ interacting agents, and we explore the impact of the model parameters. In particular, we fix $\alpha=0.6$ and we consider two scenarios, with $\beta=0.1$ and $\beta=1$, respectively. 

A simulation of the opinion-leadership model in~\eqref{eq:leader-model} in the first scenario is illustrated in Figs.~\ref{fig:examplea}--\ref{fig:exampleb}. In this case, we observe that the order of the opinions is preserved throughout the evolution. In particular, agent 1 maintains the strongest negative opinion, while agent 6 remains the one with the strongest positive opinion. Agent 1 achieves a strong leadership that forces the positive opinion of agents 4,6 and 8 to reduce, that is also the reason why those are the ones with less influence levels at the equilibrium.
Because $\alpha=0.6$ is relatively large compared to $\beta=0.1$, the leadership $y_i$ evolves faster than the opinions, allowing agents to adjust their influence levels quickly while the opinions remain largely stable.

The second scenario is illustrated in Figs.~\ref{fig:examplec}--\ref{fig:exampled}. Interestingly, we observe that here the order of opinions changes over time. For instance, agent 1 starts with the most negative opinion but ends up less extreme, while agents initially with moderately positive opinions (e.g., agents 2, 6, and 8) become slightly more extreme and they emerge as leaders in the social network. In this case, due to the relatively large value of $\beta=1$ compared to $\alpha=0.6$, the second term in the leadership dynamics (which penalizes disagreement) strongly affects $y_i$, causing rapid adjustments. As a result, agents with moderate opinions achieve higher levels of leadership.


\section{Equilibria and Stability Analysis}\label{sec:stability}
In this section, we investigate the existence of equilibria for the opinion-leadership model in \eqref{eq:leader-model} and their stability.

We first establish that any equilibrium point $x^*$, if it exists,  is necessarily bounded by the initial condition of the opinions. More precisely, given $\underline{x} = \min_j x_j(0)$ and $\bar x = \max_j x_j(0)$, $x^*$ lies in the hyper-rectangle $\mc X:= [\underline{x},\, \bar x ]^n$.

\begin{prop}\label{prop:K} Let $x^*$ be an equilibrium of~\eqref{eq:leader-model}. Then, $x^* \in \mc X$.
\end{prop}
\begin{pf}
First, observe that
$$(I - (I-Y)W) \boldsymbol{1} = \boldsymbol{1} - W \boldsymbol{1} - Y W \boldsymbol{1} = Y \boldsymbol{1},$$
so that
$$\boldsymbol{1} = \big(I - (I- Y(x))W\big)^{-1} Y \boldsymbol{1}.$$
Using this identity together with~\eqref{eq:leader-model}, we estimate upper and lower bounds for $x^*$ as
\begin{align*}
x^* &\leq \bar{x} \, [I - (I- Y(x))W]^{-1} Y \boldsymbol{1} = \bar{x} \, \boldsymbol{1}, \\
x^* &\geq \underline x \, [I - (I- Y(x))W]^{-1} Y \boldsymbol{1} = \underline x \, \boldsymbol{1}.
\end{align*}
Hence, each component of $x^*$ satisfies
$\underline{x} \le x_i^* \le \bar x$, for all $i=1,\dots,n$,
which yields the claim.\qed
\end{pf}


After having established this general result, we provide now a more precise characterization of the equilibria of the opinion-leadership model, demonstrating that the origin is always an unstable equilibrium, and providing a sufficient condition for the convergence of the dynamics.
\begin{thm}\label{th:general-eq} 
	Consider the opinion-leadership model in \eqref{eq:leader-model} with initial condition $(x(0), y(0))$. 
		Then, the following hold:
        \begin{enumerate}
            \item[(i)] $(\boldsymbol{0},\boldsymbol{0})$ is an unstable equilibrium point;
            \item[(ii)] If $x_i(0) \neq 0$ for all $i$ and, given $\varepsilon>0$ and $M := \max\{|\underline{x}|, |\bar x|\}$, it holds
            \be \label{eq:cond-norm-infty} \frac{3 \beta (\bar x - \underline{x})^2}{\alpha \varepsilon^2}\left(\frac{2M}{\varepsilon}+1\right)<1\,,\ee
            then the dynamics converges to the unique equilibrium point $(x^*,y^*)$ that depends on the initial condition and such that $y^* > \boldsymbol{0}$ and $\min_i |x_i^*|\geq\varepsilon$.
        \end{enumerate}
\end{thm}
\begin{pf}
	(i) Note that $(x^*, y^*) \in \Delta$ is an equilibrium if and only if  
    $$\begin{cases}
	\boldsymbol{0} = - x + (I-Y)Wx + Y x(0), \\
	\boldsymbol{0}= \rho\, (I - Y) \mathrm{diag}(x)\, x - \, Y\mathrm{diag}(( I - W)x)\, ( I - W)x,
\end{cases} $$
    where we define $\rho:= \alpha/\beta$. 
    This implies that the equilibrium points consist of the trivial point $(\boldsymbol{0}, \boldsymbol{0})$ and all points $(x^*, y^*)$ with $x_i^* \neq 0$ and 
	\be \label{eq:yi*} y_i^* = \frac{\rho {x_i^*}^2}{\rho {x_i^*}^2+  (x_i^*- \sum_j W_{ij} x_j^*)^2}, \quad \forall i = 1,\dots,n.\ee
    We can now study the local stability of $(\boldsymbol{0}, \boldsymbol{0})$ using the linearization theorem. Linearizing \eqref{eq:leader-model-compact} around $(\boldsymbol{0},\boldsymbol{0})$, we find the Jacobian matrix
    $$J_{(\boldsymbol{0},\boldsymbol{0})} = \begin{bmatrix} 
    	W - I \, &\,  \mathrm{diag}(x(0)) \\[1ex]
    	\mathbb{O} & \mathbb{O}
    \end{bmatrix}.
    $$
    Since $W$ is row-stochastic, the spectrum of $J_{(\boldsymbol{0},\boldsymbol{0})}$ contains $n-1$ eigenvalues with negative real part and $n+1$ zero eigenvalues. Therefore, the linearization theorem alone does not allow us to conclude on stability and we rely on the center manifold method. The stable subspace is generated by the eigenvectors associated to the negative eigenvalues of $(W-I)$. The central subspace is instead generated by the eigenvectors associated to the null eigenvalues, which means 
    $$\mc V_{\mc C} = \{ (x,y): (W-I)x + \diag(x(0))y =\boldsymbol{0} \}\,,$$
    and the center manifold is locally approximated with the function $y = h(x)$ where $h(\boldsymbol{0}) = \boldsymbol{0}$, $\dot h(\boldsymbol{0}) = \boldsymbol{0}$ and $h(x) = o(|x|)$. Therefore, the reduced equation on the center manifold is
    $$ \dot y = \alpha (1-h(x))x^2 + o(|x|^2),$$
    for which the equilibrium $(\boldsymbol{0},\boldsymbol{0})$ is unstable. Therefore, the origin is unstable also for the original system \eqref{eq:leader-model-compact}. 
    
    (ii) We proceed with the characterization of the equilibrium points given in \eqref{eq:yi*}. Substituting \eqref{eq:yi*} into the first equation of \eqref{eq:leader-model} at the equilibrium yields
	$$ \Big(x_i^* - \sum\nolimits_j W_{ij} x_j^*\Big)^3 + \rho  {x_i^*}^2(x_i^* - x_i(0)) = 0, \, \forall i\in1,\dots,n.$$
	To characterize the equilibria, we study the corresponding fixed-point equation  
	\begin{equation}\label{eq:fixed-point} 
		x = -\frac{1}{\rho} (X^{-1})^2 \mathrm{diag}((I-W)x)^2 (I-W)x + x(0)\,,
	\end{equation}
	where $X = \mathrm{diag}(x_1, \dots, x_n)$. 
    For a given $\varepsilon>0$, define
    $\mc X_\varepsilon :=\{x\in \mc X:|x_i|\geq \varepsilon >0, \forall i\}$, 
    and let $\phi: \mc X_\varepsilon \to \mc X_\varepsilon$ denote the right-hand side of \eqref{eq:fixed-point}. Then, fixed points of $\phi$ correspond exactly to the equilibria of the system.
    
    To analyze existence and uniqueness, we compute the Jacobian of $\phi$. Its entries are
    \begin{align*}
    	\frac{\partial \phi_i(x)}{\partial x_i}\!\! &= \!- \frac{3 x_i(x_i \!-\! \sum_j W_{ij}x_j)^2(1\! -\! W_{ii}) \!-\! 2 (x_i \!-\! \sum_j W_{ij}x_j)^3}{\rho x_i^3} \\
		&= \!- \frac{ (x_i - \sum_j W_{ij}x_j)^2 (x_i(1 -3 W_{ii}) + 2 \sum_j W_{ij}x_j)}{\rho x_i^3},\\
		\frac{\partial \phi_i(x)}{\partial x_k} \!\!&= -  \frac{3 (x_i - \sum_j W_{ij}x_j)^2(-W_{ik})}{\rho x_i^2}, \quad \forall k \neq i.
	\end{align*}
	Hence, the infinity norm of the Jacobian satisfies the following inequality:
	\begin{align*}
		\left\|\frac{\partial \phi(x)}{\partial x} \right\|_\infty &= \max_i\left\{ \left|\frac{\partial \phi_i(x)}{\partial x_i}\right| + \sum_{k\neq i} \left|\frac{\partial \phi_i(x)}{\partial x_k}\right|\right\}\\
		&\leq \frac{1}{\rho}\left( \frac{6M(\bar x - \underline{x})^2}{\varepsilon^3} + \frac{3(\bar x- \underline{x})^2}{\varepsilon^2}\right)\\
		& =  \frac{3 (\bar x - \underline{x})^2}{\rho \varepsilon^2}\left(\frac{2M}{\varepsilon}+1\right),
	\end{align*}
    where the inequality follows from Proposition \ref{prop:K} and the definitions of $M$ and $\mc X_\varepsilon$.
    Under condition \eqref{eq:cond-norm-infty}, $\phi$ is therefore a contraction map. By the Banach fixed-point theorem \citep{kirk2001metric}, it admits a unique fixed point $x^*$ on the complete metric space $\mc X_\varepsilon$. Moreover, for any initial condition $x(0)$ satisfying \eqref{eq:cond-norm-infty}, the sequence of iterates $x_{n+1} = \phi(x_n)$ converges to $x^*(x(0))$.
    Since $x^*(x(0))$ is a fixed point of $\phi$, it corresponds to an equilibrium of the dynamics \eqref{eq:leader-model-compact}. Therefore, under condition \eqref{eq:cond-norm-infty}, the system converges to the equilibrium $(x^*(x(0)), y^*(x(0)))$, with $y^*(x(0))$ given by \eqref{eq:yi*}.\qed
\end{pf}


\begin{rem}
    It is worth highlighting that condition~\eqref{eq:cond-norm-infty} is only sufficient and does not characterize the exact stability boundary. The admissible region should thus be interpreted as a conservative inner approximation of the true stability region. 
\end{rem}

The results of Theorem~\ref{th:general-eq} are illustrated in Fig.~\ref{fig:admissible-region}. In particular, we illustrate the set of admissible pairs $(\varepsilon, \rho)$ satisfying~\eqref{eq:cond-norm-infty} for the case in which $\bar x = 0.1$ and $\underline{x} = -0.1$, so that $M = 0.1$.
The region ensuring existence and convergence (depicted in cyan) increases for larger values of $\rho = \alpha/\beta$, meaning that a weaker sensitivity to disagreement enables equilibrium convergence even when opinions remain relatively close to the origin (small $\varepsilon$).

\begin{figure}
    \centering
    \includegraphics[width=0.65\linewidth]{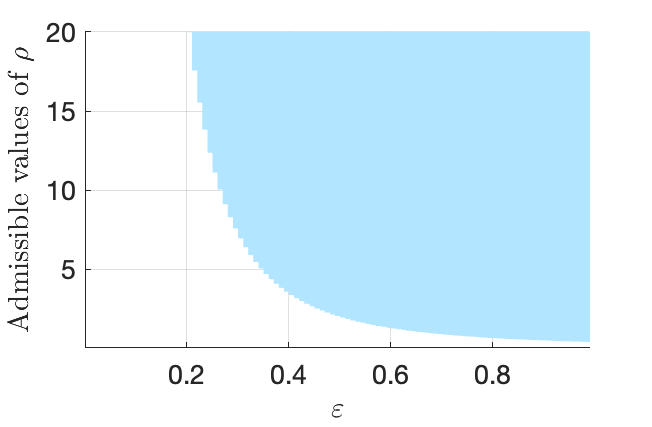}
    \caption{Region of admissible parameter values $(\varepsilon, \rho)$ for which
condition~\eqref{eq:cond-norm-infty} holds true.}
    \label{fig:admissible-region}
\end{figure}

While Theorem~\ref{th:general-eq} provides a sufficient condition for convergence to an equilibrium, it does not allow, in general, to establish a closed-form expression for such an equilibrium. Now we derive a corollary showing that, in the simple scenario in which agents share the same initial opinion, a complete characterization of the behavior of the opinion-leadership model is possible.
\begin{cor}\label{cor:equal}
    Consider the opinion-leadership model in \eqref{eq:leader-model}. If $x(0) = x_0 \boldsymbol{1}$ with $x_0 \neq 0$, then the system converges to the unique equilibrium $(x^* ,\boldsymbol{1})=( x_0 \boldsymbol{1},\boldsymbol{1})$.
\end{cor}
\begin{pf}
    Observe that for each agent $i$,
	$$\sum\nolimits_{j} W_{ij} x_j(0) = x_0 \sum\nolimits_{j} W_{ij} = x_0\,,$$
	since $W$ is row-stochastic. Hence, $x_i(0) - \sum_j W_{ij} x_j(0) = 0$, and the dynamics for $x_i$ becomes
	$\dot{x}_i(0) = 0$. But then, this can be applied for all $t\geq 0$ since $x_i (t)= x_i(0) = x_0$ and also $x_i(t) - \sum_j W_{ij} x_j(t) = 0$ for all $t \geq 0$. Similarly, the dynamics for $y_i$ reduces to
	$\dot{y}_i = \rho (1 - y_i) x_0^2$, 
	which vanishes only for $y_i = 1$ (since $x_0 \neq 0$).  
	Therefore, $(x^*,y^*) = (x_0 \boldsymbol{1}, \boldsymbol{1})$ is the unique equilibrium. Note that the condition in \eqref{eq:cond-norm-infty} is satisfied in this case, since $\underline{x}$ coincides with $\bar x$. Then the system will converge to the equilibrium $(x_0 \mathbf{1}, \boldsymbol{1})$.\qed
\end{pf}
In plain words, Corollary~\ref{cor:equal}, all agents maintain the same opinion over time, and every agent reaches the maximum level of leadership, reflecting the complete alignment of opinions within the network.


\section{Time-scale Separation}\label{sec:timescale}
In many social and organizational contexts, the dynamics of opinions and leadership may evolve at a different pace. For instance, in offline social communities, individuals may quickly adapt their expressed opinions in response to peer pressure, while the process of gaining or losing leadership status typically develops more slowly, depending on reputation, credibility, or long-term visibility~\citep{Nakayama2019}. Conversely, in settings such as online platforms, the rise of leadership or influence (e.g., becoming a trend-setter or influencer) may occur very rapidly, while opinions on specific issues remain relatively stable~\citep{Onnela2010}. 

To capture these two scenarios, we introduce a \emph{time-scale separation} between the opinion and leadership dynamics. We rewrite the opinion-leadership dynamics as
\be \label{eq:leader-model-timescale}
\begin{cases}
	\dot{x}_i = \ds - \Big( x_i - \sum\nolimits_{j} W_{ij} x_j \Big)(1 - y_i) - y_i\big(x_i - x_i(0)\big) \\
	\tau \dot{y}_i = \ds \alpha(1-y_i)x_i^2 - \beta y_i \Big( x_i - \sum\nolimits_{j} W_{ij} x_j \Big)^2 ,
\end{cases}
\ee
so that the parameter $\tau >0$ acts as a timescale separation variable \citep{berglund2006noise}, controlling the relative speed of the leadership dynamics with respect to opinion adaptation.  
Depending on the value of $\tau$, the system exhibits two distinct limiting regimes. When $\tau \to 0$, leadership adjusts rapidly relative to opinions, leading to the \emph{fast-leadership regime}, which is representative, e.g., of the online social network scenario described in the above. When $\tau \to \infty$, instead, leadership evolves much more slowly than opinions, giving rise to the \emph{slow-leadership regime}, representative of the offline context discussed in the previous paragraph. 

This distinction allows the model to capture a wide range of dynamical behaviors across different social or organizational settings. Moreover, as we shall extensively discuss in the rest of this section, such a time-scale separation allows us to perform a deeper analytical study of the system, gaining analytical insights into the emergent leadership behavior of the population. 

\subsection{Fast-leadership regime}
In the fast-leadership regime (i.e., in the limit $\tau \to 0$), the leadership variable $y$ evolves on a much faster time scale than the opinions $x$, which can be initially assumed to be constant, yielding the following result.

\begin{prop}\label{prop:fast}
    Consider the opinion-leadership in \eqref{eq:leader-model-timescale} under the fast-leadership regime ($\tau \to 0$). Given $x$ constant, the leadership of each node $y_i$ will converge to 
    \begin{equation} \label{eq:y-equilibrium}
	y_i^*(x) = 
	\begin{cases}
		0 & \text{if } x_i = 0,\\[2mm]
		\ds \frac{ \rho x_i^2}{ \rho x_i^2 +  \, (x_i - \sum_{j} W_{ij} x_j )^2} & \text{if } x_i \neq 0.
	\end{cases}
\end{equation}
for all $i\in1,\dots,n$.
\end{prop}
\begin{pf}
   Over the time scale on which $y$ evolves, we can assume $x$ to be quasi-constant. In this approximation, the second equation of \eqref{eq:leader-model-timescale} is linear in $y_i$ and solving for the equilibrium $\dot{y}_i = 0$, we obtain \eqref{eq:y-equilibrium}. With similar considerations as in the proof of Theorem~\ref{th:general-eq}, we get that the equilibrium $(\boldsymbol{0}, \boldsymbol{0)}$ is unstable.
   Focusing now on the equilibrium points where $y_i^*(x) > 0$, stability can be assessed by examining the eigenvalues of the Jacobian matrix, that is 
    $$J_y = -\mathrm{diag}\Big(\rho x_i^2 +   \Big(x_i - \sum\nolimits_{j} W_{ij} x_j \Big)^2\Big).$$
    Since $J_y$ is diagonal with strictly negative entries, all its eigenvalues are negative. By standard results from linear system theory, this implies that each $y_i$ converges to its equilibrium $y_i^*$ for fixed $x$. \qed
\end{pf}

Once the leadership variable $y$ has converged to its equilibrium manifold computed in Proposition~\ref{prop:fast}, according to the time-scale separation argument~\citep{berglund2006noise}, the opinion evolves according to the reduced dynamics  
\be\label{eq:dotx-fast}\dot{x}_i = -\frac{(x_i - \sum_j W_{ij} x_j)^3 + \rho x_i^2 (x_i-x_i(0))}{\rho x_i^2 +  (x_i - \sum_j W_{ij} x_j)^2}.\ee
Due to lack of space, we do not report here the full proof on the characterization of the equilibria, but we provide a brief outline.
To characterize the equilibria, we can study the corresponding fixed-point equation, as defined in \eqref{eq:fixed-point}.
Then we can prove existence and stability of an equilibrium $(x^*, y^*)$ as done in Proposition \ref{th:general-eq}(ii).

Figure~\ref{fig:fastLead} illustrates the behavior of the opinion-leadership model \eqref{eq:leader-model-timescale} under a fast-leadership time-scale regime. In this case, for all $i$, the leadership variables $y_i$ quickly converge to the manifold defined by \eqref{eq:y-equilibrium}, after which the opinion dynamics $x_i(t)$ evolve more slowly, eventually reaching an equilibrium. 
\begin{figure}
    \subfloat[Opinions]{\includegraphics[width=0.53\linewidth]{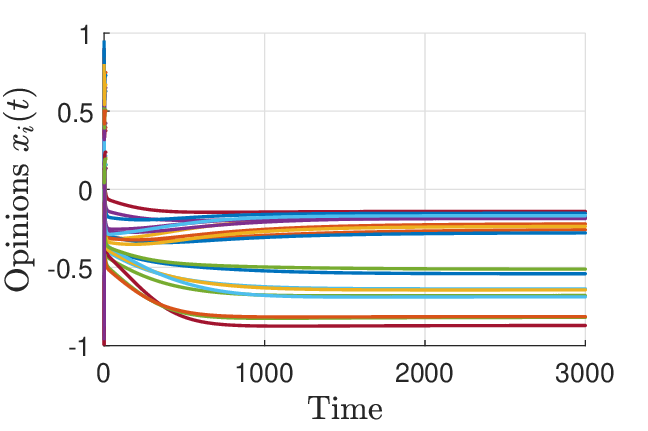}}
    \subfloat[Leadership]{\hspace{-5pt}\includegraphics[width=0.53\linewidth]{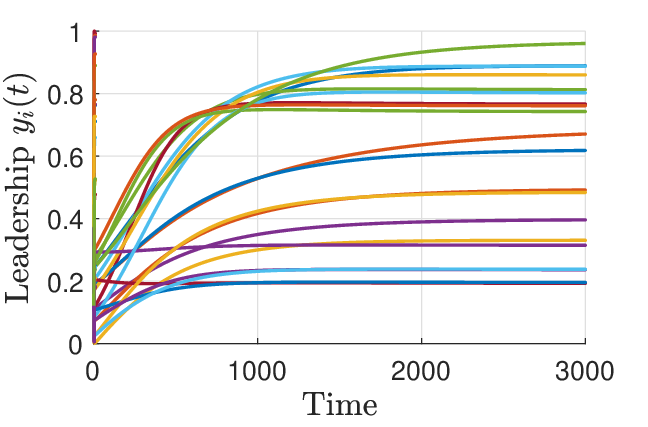}}
    \caption{Simulation of the model in \eqref{eq:leader-model-timescale} under the fast-leadership regime ($\tau \to 0$). }
    \label{fig:fastLead}
\end{figure}

\subsection{Slow-leadership regime} 
In the slow-leadership regime (i.e., in the limit $\tau \to \infty$), the evolution of the opinion variable $x$ occurs on a much faster time scale than that of leadership $y$, which can be initially assumed as constant, obtaining the following result. First, define the Laplacian matrix $L:= I-W$.

\begin{prop}\label{prop:slow}
    Consider the opinion-leadership in \eqref{eq:leader-model-timescale} under the slow-leadership regime ($\tau \to \infty$). Given $y$ constant, the opinion will converge to 
    \be\label{eq:x-equilibrium}x^*(y) = \big(L+ Y W\big)^{-1} Y x(0),\ee
    where $Y = \mathrm{diag}(y_1, \dots, y_n)$.
\end{prop}
\begin{pf}
    Over the time scale on which $x$ evolves, we can treat $y$ as quasi-constant. This allows us to analyze the opinion dynamics separately, considering $y$ fixed.
    In this limit, the equilibrium of the opinion dynamics $\dot{x} = 0$ is given by \eqref{eq:x-equilibrium}.
    The convergence to $x^*(y)$ follows from the fact that, for a fixed $y$, the opinion dynamics defines a linear system and equilibrium stability can be assessed by examining the eigenvalues of the Jacobian matrix, that is
    $$J_x = -\big((I-Y)L + Y\big).$$
    Since $W$ is row-stochastic, all eigenvalues of $J_x$ are negative. By standard results from linear system theory, this implies that $x$ converges exponentially to the unique equilibrium $x^*(y)$ for fixed $y$. \qed
\end{pf}

\begin{rem}
   The expression from Proposition~\ref{prop:slow} can be interpreted as a modification of the equilibrium of a  Friedkin--Johnsen opinion dynamics model~\citep{Friedkin1990}, where each agent’s final opinion is a convex combination of their initial opinion and the weighted average of their neighbors’ opinions, with weights modulated by their current leadership $y_i$.
\end{rem}

According to the time-scale separation principle~\citep{berglund2006noise}, once the fast variable $x$ has reached its equilibrium manifold $x^*(y)$, the slow leadership variable $y$ evolves according to the reduced dynamics
\begin{align}\label{eq:dotx-slow}
    \dot{y} =& \alpha (I \! -\! Y)\diag(x^*(y)) x^*(y)\!-\! \beta Y\diag(Lx^*(y)) Lx^*(y).
\end{align}
This is illustrated in Fig.~\ref{fig:slowLead}, which shows how opinions $x_i(t)$ rapidly settle into a quasi-steady configuration, while the leadership variables $y_i$ adjust on a much slower timescale. In this scenario, the opinion dynamics effectively react to a nearly constant leadership.

\begin{figure}
    \subfloat[Opinions]{\includegraphics[width=0.54\linewidth]{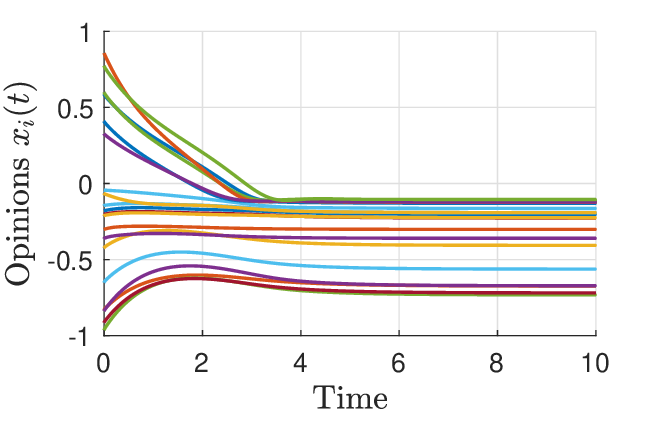}}
    \subfloat[Leadership]{\hspace{-5pt}\includegraphics[width=0.54\linewidth]{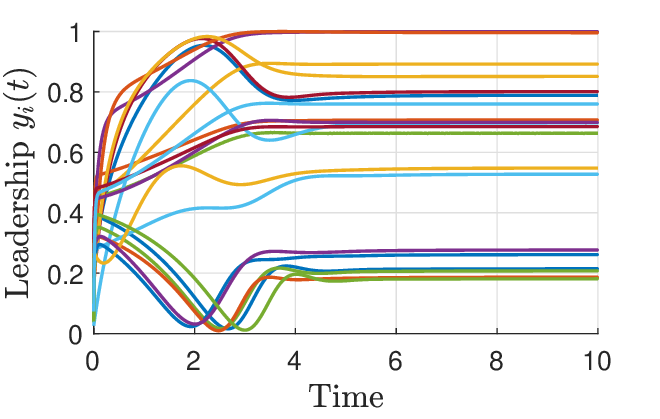}}
    \caption{Simulation of the model in \eqref{eq:leader-model-timescale} under the slow-leadership regime ($\tau \to \infty$).   }
    \label{fig:slowLead}
\end{figure}

Finally, in the context of a slow-leadership time-scale, the following result provides a sufficient condition for the existence and uniqueness of an equilibrium point for the system \eqref{eq:dotx-slow}.
\begin{prop}\label{prop:slow2}
	Consider the reduced dynamics in \eqref{eq:dotx-slow} in the slow-leadership regime ($\tau\to\infty$), and let $x_i(0) \neq 0$ for all $i$. If 
	\be\label{cond2} 4 \frac{M^3}{\underline{x}^2} \frac{\rho \bar{x}^2 + ( \bar x- \underline{x})^2}{\rho \underline{x}^2 + (\bar x- \underline{x})^2}\left( \frac{M(2+\rho)}{\rho \underline{x}^2 + ( \bar x- \underline{x})^2} + 1\right) <1, \ee
    then there exists a unique equilibrium $y^*$ and the dynamics will converge to $y^*$ for any initial condition $y(0)$. 
\end{prop}
\begin{pf}
Define the implicit map $\varphi: [0,1] \to [0,1]$ such that at the equilibrium we can study the following fixed-point equation
\be\label{eq:varphi} y = \varphi(y) := \Big(\rho X^2+  \big(\mathrm{diag}(Lx^*(y))\big)^2\Big)^{-1} \rho X^2,\ee
where $X = \mathrm{diag}(x^*(y))$.
Note that from Proposition \ref{prop:K} and \eqref{eq:x-equilibrium}, it follows that, given the set
$$\mc X_{\varphi}:= \Bigg[ \frac{\rho \underline{x}^2}{\rho \bar x^2 + (\bar x-\underline{x})^2},1\Bigg]^n,$$
$\phi\left(\mc X_{\varphi}\right) \subseteq \mc X_{\varphi}.$
As done in Theorem \ref{th:general-eq}, in order to analyze existence and uniqueness of fixed points of $\varphi$, which correspond exactly to equilibria of the system \eqref{eq:dotx-slow}, we compute the Jacobian matrix of $\varphi$. Defining the following diagonal matrix $A(y) :=  \rho X^2+ \big(\mathrm{diag}(Lx^*(y))\big)^2$, we get 
{ \thinmuskip=0.2mu  
	\medmuskip=0.2mu   
	\begin{align*}
		\mspace{-5mu}&\frac{\partial \varphi}{\partial y}(y) =- \rho  A^{-1}(y) \frac{\partial A}{\partial y}(y) A^{-1}(y)  X^2  +2 \rho  A^{-1}(y) X \frac{\partial x^*}{\partial y}(y),
\end{align*}}
where 
{\thinmuskip=-0.2mu  \medmuskip=-0.2mu   
\begin{align*}
    \frac{\partial A}{\partial y}(y) &= 2\Big(\rho X + \mathrm{diag}(Lx^*(y)) L\Big) \frac{\partial x^*}{\partial y} (y),
\end{align*}}
and 
{\thinmuskip=0mu  \medmuskip=0mu   
    \begin{align}
    \mspace{-10mu} \frac{\partial x^*}{\partial y}(y) &= \big(L+ YW\big)^{-1} \left(\mspace{-5mu}- W \big(L+YW\big)^{-1} Y X(0) + X(0)\mspace{-5mu}\right) \nonumber\\
    &= \big(L+ Y W\big)^{-1}\mathrm{diag}\big(x(0) - W x^*(y)\big).
\end{align}}

Recalling that $W$ is row-stochastic, it follows that the infinity norm of $\big(L+ Y W\big)^{-1} $ satisfies  
\be \begin{aligned}\label{eq:norm-inf-inv}
	\big\| (L+ Y W)^{-1} \big\|_\infty &\leq \frac{1}{1-\left\|(I-Y)W\right\|_\infty } \\
	&= \frac{1}{1-(1-\min_i y_i)} \\
    &= \frac{\rho \bar x^2 + (\bar x-\underline{x})^2}{\rho \underline{x}^2}.
\end{aligned}\ee
We can now bound the infinity norm of the Jacobian matrix of $\varphi$ as follows, where for brevity, we omit some intermediate steps,
\begin{align*}
	\mspace{-5mu}\left\| \frac{\partial \varphi(y)}{\partial y} \right\|_\infty \mspace{-10mu}
	&\leq 4 \frac{M^3}{\underline{x}^2} \frac{\rho \bar{x}^2 + ( \bar x- \underline{x})^2}{\rho \underline{x}^2 + (  \bar x- \underline{x})^2}\left( \frac{M(2+\rho)}{\rho \underline{x}^2 + (  \bar x- \underline{x})^2} + 1\right)\! .
\end{align*}
Since \eqref{cond2} holds true, $\varphi$ is a contraction map. By the Banach fixed-point theorem \citep{kirk2001metric}, the contraction mapping $\varphi$ admits a unique fixed point $y^*$ in $\mc X_{\phi}$, and the sequence of iterates $y_{n+1} = \varphi(y_n)$ converges to $y^*$ for any initial condition $y(0)$. Since $y^*$ is a fixed point of $\varphi$, it corresponds to an equilibrium of the reduced dynamics \eqref{eq:dotx-slow}. Therefore, under condition \eqref{cond2}, the reduced system converges to the equilibrium point.\qed
\end{pf}

The assumption in Proposition~\ref{prop:slow2} is illustrated in Fig.~\ref{fig:admissible-region2}, where we show the set of admissible pairs $(\bar x, \rho)$ satisfying~\eqref{cond2} for the case in which $\underline x = -1$.
The plot illustrates how the (conservative) region ensuring existence and convergence of the dynamics in the slow-leadership regime increases as the ratio $\rho = \alpha/\beta$ increases, suggesting that convergence is easier to achieve when strong-opinionated individuals are more likely to emerge as leaders.

\begin{figure}
    \centering
    \includegraphics[width=0.65\linewidth]{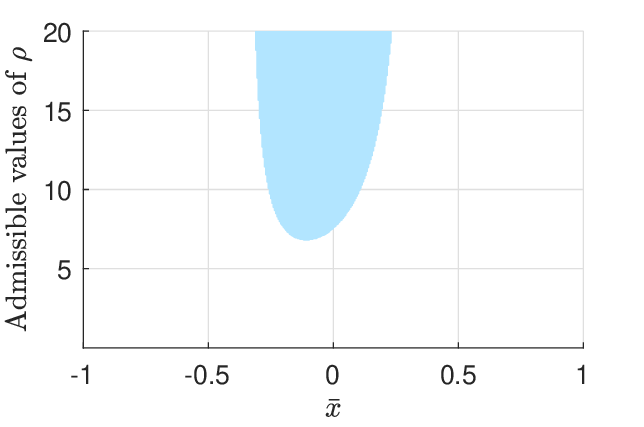}
    \caption{Region of admissible parameter values $(\bar x, \rho)$ for which
condition~\eqref{cond2} holds true.}
    \label{fig:admissible-region2}
\end{figure}

\section{Conclusion}\label{sec:conclusion}
This work introduces a novel modeling framework that captures the endogenous emergence of leadership in social networks through the coupled evolution of opinions. By embedding a dynamic leadership mechanism within a Friedkin-–Johnsen inspired opinion model, we emphasized a fundamental trade-off observed in social systems: strong and confidently expressed convictions can promote leadership emergence, yet excessive deviation from the prevailing social context tends to weaken influence due to reduced validation from peers.

We formally analyzed the resulting nonlinear dynamics and established sufficient conditions ensuring convergence to a nontrivial equilibrium point. In addition, we examined two relevant time-scale separation regimes, showing how the relative speed of opinion and leadership evolution affects the long-term social configuration.

The proposed framework offers several opportunities for further research. From a modeling perspective, one may consider alternative functional forms for the leadership dynamics, or simplifying assumptions such as restricting individual opinions to a positive domain, to derive stronger analytical guarantees. On the applied side, incorporating heterogeneous interaction patterns extracted from empirical data would allow validating the model and testing its predictive power in real social systems.
\bibliography{leaderness_bib.bib}

\end{document}